\documentclass[aps,pra,reprint,superscriptaddress,nofootinbib]{revtex4-2}
\usepackage[utf8]{inputenc}
\usepackage{physics}
\usepackage{bbold}
\usepackage{graphicx}
\usepackage{multirow}
\usepackage{hyperref}
\hypersetup{colorlinks=true, citecolor=red}
\usepackage{cleveref}
\usepackage{xcolor}
\usepackage{bbold} 
\usepackage{tabularx} 
\usepackage{soul}  
\usepackage{tikz}

\usepackage[caption = false]{subfig}

\definecolor{pink}{HTML}{F282B4}
\definecolor{corn}{HTML}{6495ED}

\newcommand{\KL}{D_{\text{KL}}}

\newcommand{\Hv}{H_{\mathrm{v}}}
\newcommand{\Hh}{H_{\mathrm{h}}}
\newcommand{\Hint}{H_{\mathrm{int}}}
\newcommand{\Wv}{\mathcal{W}_{\mathrm{v}}}
\newcommand{\Wh}{\mathcal{W}_{\mathrm{h}}}
\newcommand{\Wint}{\mathcal{W}_{\mathrm{int}}}
\newcommand{\RBM}{\mathrm{RBM}}
\newcommand{\sqRBM}{\mathrm{sqRBM}}
\newcommand{\sqBM}{\mathrm{sqBM}}
\newcommand{\QRBM}{\mathrm{QRBM}}
\newcommand{\TVD}{\mathrm{TVD}}

\usepackage{makecell} 
\usepackage{mathtools}
\usepackage{tabularx} 
\usepackage{amsmath}
\usepackage{amssymb}
\usepackage{amsthm}
\usepackage{flushend}

\newcolumntype{C}[1]{>{\centering\arraybackslash}p{#1}}

\newtheorem{theorem}{Theorem}
\newtheorem{definition}{Definition}
\newtheorem{lemma}{Lemma}
\newtheorem{proposition}{Proposition}

\newtheorem{corollary}{Corollary}

\begin{document}

\title{Expressive equivalence of classical and quantum restricted Boltzmann machines}

\date{\today}

\author{Maria Demidik}
\email{maria.demidik@desy.de}
\affiliation{Deutsches Elektronen-Synchrotron DESY, Platanenallee 6, 15738 Zeuthen, Germany}
\affiliation{Computation-Based Science and Technology Research Center, The Cyprus Institute, 20 Kavafi Street, 2121 Nicosia, Cyprus}

\author{Cenk T\"uys\"uz}
\affiliation{Deutsches Elektronen-Synchrotron DESY, Platanenallee 6, 15738 Zeuthen, Germany}
\affiliation{Institut für Physik, Humboldt-Universit\"at zu Berlin, Newtonstr. 15, 12489 Berlin, Germany}

\author{Nico Piatkowski}
\affiliation{Fraunhofer IAIS, Schloss Birlinghoven, 53757, Sankt Augustin, Germany}

\author{Michele Grossi}
\affiliation{European Organisation for Nuclear Research (CERN), Espl. des Particules 1211 Geneva 23, Switzerland}

\author{Karl Jansen}
\affiliation{Computation-Based Science and Technology Research Center, The Cyprus Institute, 20 Kavafi Street, 2121 Nicosia, Cyprus}
\affiliation{Deutsches Elektronen-Synchrotron DESY, Platanenallee 6, 15738 Zeuthen, Germany}

\begin{abstract}
Quantum computers offer the potential for efficiently sampling from complex probability distributions, attracting increasing interest in generative modeling within quantum machine learning. This surge in interest has driven the development of numerous generative quantum models, yet their trainability and scalability remain significant challenges. A notable example is a quantum restricted Boltzmann machine (QRBM), which is based on the Gibbs state of a parameterized non-commuting Hamiltonian. While QRBMs are expressive, their non-commuting Hamiltonians make gradient evaluation computationally demanding, even on fault-tolerant quantum computers. In this work, we propose a semi-quantum restricted Boltzmann machine (sqRBM), a model designed for classical data that mitigates the challenges associated with previous QRBM proposals. The sqRBM Hamiltonian is commuting in the visible subspace while remaining non-commuting in the hidden subspace. This structure allows us to derive closed-form expressions for both output probabilities and gradients. Leveraging these analytical results, we demonstrate that sqRBMs share a close relationship with classical restricted Boltzmann machines (RBM). Our theoretical analysis predicts that, to learn a given probability distribution, an RBM requires three times as many hidden units as an sqRBM, while both models have the same total number of parameters. We validate these findings through numerical simulations involving up to 100 units. Our results suggest that sqRBMs could enable practical quantum machine learning applications in the near future by significantly reducing quantum resource requirements. 
\end{abstract}

\maketitle
\section{Introduction}

\begin{figure*}[!t]
    \centering
    \includegraphics[width=\linewidth]{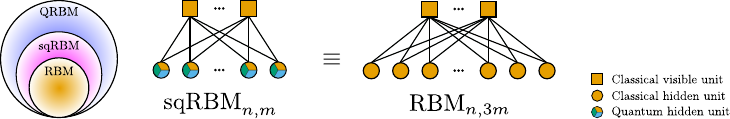}
    \caption{\textbf{Summary of main results.} This work introduces semi-quantum restricted Boltzmann machines (sqRBM) as an intermediate model, satisfying the relation $\QRBM \supseteq \sqRBM \supseteq \RBM$. sqRBMs generalize RBMs by rendering the hidden units \textit{quantum} through the use of non-commuting Hamiltonians. In Theorem~\ref{theorem:equivalance}, we show that $\sqRBM_{n,m} \equiv \RBM_{n,3m}$, where $n$ and $m$ denote the number of visible and hidden units, respectively, with both models having the same number of parameters. In pedestrian terms, RBMs require three times as many hidden units as sqRBMs to learn the same target distribution. \looseness-1}
    \label{fig:summary}
\end{figure*}

Boltzmann machines (BM) are a prominent example of machine learning models based on statistical physics~\cite{smolensky_bm, ackley_learning_1985}. They have been widely applied in areas such as collaborative filtering, dimensionality reduction, pattern recognition and generative modeling~\cite{salakhutdinov_restricted_2007, hinton_reducing_2006, salakhutdinov_deep_2009, yichuan_tang_robust_2012}. A BM is defined with a system of binary visible and hidden units. Visible units serve as input and output of the model, while hidden units represent latent variables. Overall, the model forms a classical Ising model, which can be defined through a parameterized commuting Hamiltonian.

The representational power of BMs depends on the number of hidden units~\cite{le_roux_representational_2008} and their connectivity~\cite{srbm}. Each training iteration involves Gibbs state preparation to estimate the gradients. This makes training BMs computationally demanding. To address this, a common approach is to restrict connectivity, resulting in restricted Boltzmann machines (RBM), which allow only visible-hidden connections. Despite this simplification, training remains computationally expensive~\cite{bm_hard_simulate}.

Contrastive divergence (CD) is a widely used approximation method to accelerate training of RBMs~\cite{bm_cd}. However, CD provides only a rough estimate of the true gradients, often leading to unstable convergence and limiting the practical applicability of RBMs~\cite{carreira-perpinan_contrastive_2005}. While several improvements have been proposed~\cite{bengio_justifying_2009, salakhutdinov_efficient_2012}, efficient training of BMs and RBMs remains an open challenge in machine learning research.

Quantum computing offers opportunities to facilitate the training of BMs. Researchers have proposed multiple polynomial scaling quantum algorithms for Gibbs state preparation~\cite{chen2023quantumthermalstatepreparation}. Utilizing quantum hardware for Gibbs state preparation could not only improve the training process but also offer a sampling advantage for BMs. Consequently, quantum computing could significantly increase the practical relevance of BMs. Moreover, the ability to prepare a Gibbs state on quantum computers enables generalizing the Hamiltonian of BMs with non-commuting terms, potentially enhancing the model's representational power. A model defined by a non-commuting Hamiltonian, known as a quantum Boltzmann machine (QBM), provides a framework for engineering its Hamiltonian.

Quantum machine learning aims to enhance the capabilities of machine learning models with access to quantum computers. In the domain of generative modeling, the majority of proposals in the field have been based on parametrized quantum circuits such as quantum generative adversarial networks~\cite{qgan} or quantum circuit Born machines~\cite{benedetti_generative_2019}. Recent studies have revealed that these models encounter trainability issues, such as barren plateaus~\cite{mcclean_barren_2018, rudolph_trainability_2024}, rendering them impractical. Although certain QBM constructions are similarly affected by these limitations~\cite{ent-bp}, evidence suggests that alternative QBM formulations can successfully circumvent such issues~\cite{coopmans_sample_2024}.


Besides trainability issues, a major challenge for QBMs is the gradient estimation. The gradients of a QBM, which is defined by a generic non-commuting Hamiltonian, are known to be computationally intractable~\cite{amin_quantum_2018}. To overcome this challenge, various frameworks impose constraints on the Hamiltonian~\cite{amin_quantum_2018, anschuetz2019realizingquantumboltzmannmachines, kieferova_tomography_2017}. These approaches commonly avoid training non-commuting terms, treating them instead as hyperparameters. Although this constraint makes training cheaper, it inherently restricts the model's representational power. Consequently, the computational cost of QBM training is fundamentally tied to the choice of Hamiltonian, highlighting a trade-off between tractability and expressiveness.

In this work, we propose semi-quantum restricted Boltzmann machines (sqRBM) designed for efficient gradient computation while enabling the training of non-commuting terms. This is achieved by defining a Hamiltonian that is diagonal in the subspace of visible units, while containing non-commuting terms in the subspace of hidden units. In this way, it serves as an intermediate model between RBMs and quantum restricted Boltzmann machines (QRBM) such that $\QRBM \supseteq \sqRBM \supseteq \RBM$ as illustrated in Figure~\ref{fig:summary}. Being diagonal in the subspace of visible units allows sqRBM to provide a framework to explore the impact of non-commuting terms on learning from classical data. 

In order to investigate the practical importance of sqRBMs, we establish a direct correspondence between RBM and sqRBM. In particular, Theorem~\ref{theorem:equivalance} shows the expressivity equivalence between $\sqRBM_{n, m}$ and $\RBM_{n, 3m}$, where $n$ and $m$ are the number of visible and hidden units correspondingly. This implies that an sqRBM requires only one-third of the hidden units compared to an RBM to achieve the same representational power. Additionally, we present numerical results on various datasets and models up to 100 units, further validating our theoretical findings.

While expressive equivalence between an sqRBM and an RBM suggests similar learning performance on a given task, the sqRBM’s requirement for fewer hidden units may reduce the computational costs of training and sampling, given access to a fault-tolerant quantum computer. This opens new avenues in practical applications of QBMs. An illustration of the summary of main results is provided in Figure~\ref{fig:summary}. 

The remainder of this work is structured as follows: In Section~\ref{sec:background}, we provide the necessary definitions along with a brief overview of recent developments and approaches related to QBMs. In Section~\ref{sec:results}, we begin with the theoretical intuition behind sqRBMs, followed by a presentation of the analytical output probabilities and model gradients. It concludes with numerical results obtained by simulating RBMs and sqRBMs with up to 100 units. Finally, in Section~\ref{sec:discussion}, we discuss the implications of our results for the future of quantum machine learning.

\section{Background}
\label{sec:background}

\subsection{Model definitions}
A Boltzmann machine (BM) can be described by a Hamiltonian $H$ and its corresponding Gibbs state $\rho$. BMs consist of two types of units: visible and hidden. Visible units are the ones that are observed and used for input/output purposes, while hidden units form the latent dimension, giving the model its representation power. We denote the number of visible units with $n$ and the number of hidden units with $m$. 

The Hamiltonian that describes BMs can be written as a sum of three terms as follows:
\begin{equation}
    H = \Hv + \Hh + \Hint,
    \label{eq:hamiltonian}
\end{equation}
where $\Hv$ and $\Hh$ act on visible and hidden units respectively, while $\Hint$ represents the interaction between visible and hidden units. Consequently, the corresponding Gibbs state $\rho$ of $H$ is given as
\begin{equation}
    \rho = e^{- \beta H}/\mathcal{Z},\quad \mathcal{Z} = \Tr[e^{- \beta H}],
\label{eq:gibbs_state}   
\end{equation}
where $\mathcal{Z}$ is the partition function that ensures normalization ($\Tr[\rho]=1$) and $\beta$ is the inverse temperature, which we set as $\beta=1$ to simplify subsequent equations. In this work, we focus on restricted BM configurations, allowing interactions only between visible and hidden units. We denote the number of visible units with $n$ and the number of hidden units with $m$. A visualization of the restricted BM (RBM) configuration can be seen in Figure~\ref{fig:rbm-conn}.

\begin{figure}[!h]
    \centering
    \includegraphics[width=.9\linewidth]{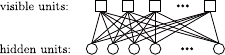}
    \caption{\textbf{Connectivity graph of restricted Boltzmann machines (RBM).} An RBM model has connections only between visible and hidden units. Lateral connections (e.g. visible to visible) are not permitted.}
    \label{fig:rbm-conn}
\end{figure}

The probability distribution of the model, denoted with $p$, can be obtained by marginalizing over hidden units, such that
\begin{equation}
    p_v = \Tr[\Lambda_v \rho], 
\label{eq:state_v}
\end{equation}
where $v \in \{0,1\}^n$ is a length $n$ bitstring and $\Lambda_v$ is a projective measurement with respect to the computational basis of the visible units given as
\begin{equation}
    \Lambda_{v} = \ket{v}\bra{v} \otimes \mathbb{1}_{2^m, 2^m}.
    \label{eq:projector-comp-basis}
\end{equation}

Of particular interest, let us define a Pauli string of length $n+m$ as the tensor product of operators from the set of Pauli matrices including the identity $\{ I, X, Y, Z \}$. A $k$-body operator acts on $k$ many qubits non-trivially, meaning it has $k$ many non-identity operators in the Pauli string representation. We write the $n+m$-qubit Pauli string of the Pauli-$Z$ operator acting on the $i$-th qubit as an example
\begin{equation}
\sigma_i^{Z} = \underbrace{I \otimes \cdots \otimes I }_{i-1} \otimes Z \otimes \underbrace{I \otimes \cdots \otimes I}_{n+m-i}.   
\end{equation}
With these definitions at hand, we provide a formal definition of a classical restricted Boltzmann machine (RBM) as follows:
\begin{definition}[Restricted Boltzmann machine (RBM)]
\label{def-rbm}
A restricted Boltzmann machine with $n$ visible and $m$ hidden units, denoted as $\RBM_{n,m}$, is described by a parameterized Hamiltonian according to Eq.~\eqref{eq:hamiltonian}. The three terms of the Hamilton are defined as follows:
\begin{align}
&\Hv = \sum_{i=1}^{n} a_{i}^{Z} \sigma_i^{Z}, \quad
\Hh = \sum_{j=1}^{m} b_{j}^{Z} \sigma_{n+j}^{Z},  \nonumber \\
&\Hint = \sum_{i=1}^{n} \sum_{j=1}^{m} w_{i, j}^{Z,Z} \sigma_{i}^{Z} \otimes \sigma_{n+j}^{Z},
\end{align}
where $\boldsymbol{a} \in \mathbb{R}^n$, $\boldsymbol{b} \in \mathbb{R}^m$ and $\boldsymbol{w} \in \mathbb{R}^{nm}$ are the parameter vectors of the model.
\end{definition}

The Hamiltonian of an $\RBM_{n,m}$ corresponds to a classical Ising model and contains only commuting operators ($\forall (i,j), \quad \left[ \sigma_i^{Z}, \sigma_{i}^{Z} \otimes \sigma_{n+j}^{Z} \right] = 0$). An $\RBM_{n,m}$ can be extended to a quantum model by incorporating non-commuting operators in the Hamiltonian. We refer to a generic model with a non-commuting Hamiltonian as a quantum restricted Boltzmann machine (QRBM).

\begin{definition}[Quantum restricted Boltzmann machine (QRBM)] \label{def-qrbm} 
A quantum restricted Boltzmann machine with $n$ visible and $m$ hidden units, denoted as $\QRBM_{n,m}$, is described by a parameterized Hamiltonian according to Eq.~\eqref{eq:hamiltonian}. The three terms of the Hamilton are defined as follows:
\begin{align}
&\Hv = \sum_{P \in \Wv}^{|\Wv|} \sum_{i=1}^{n} a_{i}^{P} \sigma_i^{P}, \quad
\Hh = \sum_{P \in \Wh}^{|\Wh|} \sum_{j=1}^{m} b_{j}^{P} \sigma_{n+j}^{P},  \nonumber \\
&\Hint = \sum_{(P,Q) \in \Wint}^{|\Wint|} \sum_{i=1}^{n} \sum_{j=1}^{m} w_{i, j}^{P, Q} \sigma_{i}^{P} \otimes \sigma_{n+j}^{Q},
\end{align}
where $\boldsymbol{a} \in \mathbb{R}^{|\Wv|n}$, $\boldsymbol{b} \in \mathbb{R}^{|\Wh|m}$ and $\boldsymbol{w} \in \mathbb{R}^{|\Wint|nm}$ are the parameter vectors of the model. $\Wv$, $\Wh$ and $\Wint$ are the sets of Pauli operators that describe the model and $|\mathcal{W}|$ denotes the cardinality of set $\mathcal{W}$.
\end{definition}

A generic QRBM contains all possible one- and two-body Pauli operators when $\mathcal{W}_{\mathrm{v}}=\mathcal{W}_{\mathrm{h}}=\{ X,Y,Z\}$ and $\mathcal{W}_{\mathrm{int}}=\mathcal{W}_{\mathrm{v}} \times \mathcal{W}_{\mathrm{h}}$. A common choice in the literature is the set that corresponds to the transverse field Ising model, such that  $\mathcal{W}_{\mathrm{v}}=\mathcal{W}_{\mathrm{h}}=\{X,Z\}$ and $\mathcal{W}_{\mathrm{int}}=\{ZZ\}$~\cite{amin_quantum_2018}. Notice that the definition of QRBM also includes the RBM when $\mathcal{W}_{\mathrm{v}}=\mathcal{W}_{\mathrm{h}}=\{Z\}$ and $\mathcal{W}_{\mathrm{int}}=\{ZZ\}$.

\subsection{Training Boltzmann machines}

Boltzmann machines can be trained by minimizing the negative log-likelihood, which is defined as
\begin{equation}
    \mathcal{L} = -\sum_v q_v \log{p_v},
\label{eq:loss_function}
\end{equation}
where $p$ and $q$ are the probability distributions of the model and the target, respectively ($\sum_v q_v = \sum_v p_v = 1$). Minimizing the negative log-likelihood is equivalent to minimizing $\KL$, which is given as
\begin{align}
    \KL (q || \, p) &= \sum_v q_v \log \left( \frac{q_v}{p_v} \right) \nonumber \\
    &= -\sum_v q_v \log p_v + \sum_v q_v \log q_v.
\end{align}
Then, gradients of both RBMs and QRBMs with respect to the negative log-likelihood can be obtained using a single formula that has two parts: positive phase and negative phase. The negative phase is obtained by measuring expectation values over the Gibbs state of the model, while the positive phase requires projective measurements on the visible unit subspace. Although these terms may look different in the classical machine learning literature, we keep the naming convention.

%
\begin{proposition}[Gradients of QRBM]\label{prop:qrbm-grads} A $\QRBM_{n,m}$ can be trained to minimize the negative log-likelihood with respect to the target probability distribution $q$ using the following gradient rule: 
\begin{equation}
    \partial_{\theta_i}{\mathcal{L}} = -\sum_v q_v \left( \underbrace{\frac{\Tr[\Lambda_v \partial_{\theta_i}{e^{-H}}]}{\Tr[\Lambda_v e^{-H}]}}_{\mathrm{\mathrm{positive~phase}}} - \underbrace{\frac{\Tr[\partial_{\theta_i}{e^{-H}}]}{\Tr[e^{-H}]}}_{\mathrm{\mathrm{negative~phase}}} \right),
    \label{eq:qrbm-grad}
\end{equation}
where $\theta_i \in \boldsymbol{\theta}$ is any real-valued parameter of the model, when the Hamiltonian terms are grouped such that $H = \sum_i \theta_i H_i$ and $\boldsymbol{\theta} \in \{\boldsymbol{a}, \boldsymbol{b}, \boldsymbol{w} \}$.
\end{proposition}

The proof is provided in Appendix~\ref{app:proof-prop1}. In the case of a commuting Hamiltonian (e.g., RBM), the gradients of the negative log-likelihood with respect to the parameters take a fairly simple form. This follows the fact that if $[\partial_{\theta_i}H, H]=0$, then $\partial_{\theta_i}{e^{-H}} = e^{-H}(-H_i)$, where $H_i = \partial_{\theta_i}H$ and $H=\sum_i \theta_i H_i$. Then, the gradients can be computed by measuring the expectation values of the Hamiltonian terms on the Gibbs state of the model. Although this appears relatively straightforward, preparing the Gibbs state is still exponentially expensive with respect to the system size $n+m$ on a classical computer. For this reason, in practice, alternative methods are often employed to avoid this step. One of the most popular approaches is called contrastive divergence (CD)~\cite{bm_cd}.

In the case of a generic non-commuting Hamiltonian, the gradients become expensive to compute, which is mainly due to the derivative of the Hamiltonian not commuting with itself ($[\partial_{\theta_i}H, H] \neq 0$) and requiring $\partial_{\theta_i}{e^{-H}}$ to be explicitly computed. This makes computing the positive phase costly, even when the Gibbs state can be prepared efficiently. Therefore, a generic QRBM cannot be trained efficiently with the vanilla gradient descent approach.

Additionally, it has been shown that gradients of QRBMs with volume-law entanglement vanish exponentially in system size~\cite{ent-bp}. This means that the number of shots required to compute the expectation values increases exponentially with system size, introducing another layer of complication to utilize generic QRBM models. \looseness-1

\subsection{Existing approaches}

Computational overhead associated with gradients of non-commuting Hamiltonians has prompted the search for methods to mitigate this issue. The majority of proposed methods introduce bounds on the loss function to simplify the gradient expressions, while others modify or constrain the model definition. 

Ref.~\cite{amin_quantum_2018} has proposed an upper-bound for the negative log-likelihood. However, this approach constrains training of non-commuting terms, treating them as hyperparameters of the model. Alternatively, Refs.~\cite{amin_quantum_2018, kieferova_tomography_2017} have proposed optimizing a lower-bound of the objective function utilizing the Golden–Thompson inequality, enabling training of non-commuting terms only when the input includes non-diagonal elements. Ref.~\cite{wiebe_generative_2019} has introduced a variational upper-bound on the quantum relative entropy for QBMs restricted to commuting operators on the hidden units. Nevertheless, the expressive capacity of QBMs within the suggested framework remains an open question, particularly when applied to classical data. 

From a different perspective, Ref.~\cite{zoufal_variational_2021} has introduced a training algorithm based on the variational quantum imaginary time evolution. Although this approach enables training of generic QBMs, it encounters scalability issues of variational algorithms~\cite{mcclean_barren_2018}, thereby limiting their practical applications.

Last but not least, recent results have shown that fully-visible QBMs (models with no hidden units) can be trained sample-efficiently~\cite{Kappen_2020, coopmans_sample_2024}. Many studies have followed this result to show the capabilities of fully-visible QBMs on learning from classical and quantum data~\cite{patel2024quantumboltzmannmachinelearning, tuysuz2024learninggeneratehighdimensionaldistributions, patel2024naturalgradientparameterestimation, minervini2025evolved}. While fully-visible QBMs are more expressive than their classical counterparts, their lack of hidden units considerably limits their applicability to many practical tasks~\cite{montufar_expressive_2014}.

Existing proposals in the literature face limitations such as restricted training of non-commuting terms or reliance on non-diagonal density matrices, while classical probability distributions are represented by diagonal density matrices. These challenges hinder practical applications of QBMs for classical probability distributions. In the following section, we introduce the semi-quantum restricted Boltzmann machine, a novel framework designed for classical data that circumvents the aforementioned issues of QBMs.

\section{Results} \label{sec:results}

\subsection{Theoretical intuition}
Challenges in training generic QBMs stem from non-trivial estimation of $\partial_{\theta_i}{e^{-H}}$ and $\Tr[\Lambda_v \partial_{\theta_i}{e^{-H}}]$. In the case of a non-commuting Hamiltonian of the form $H=\sum_i \theta_i H_i$, one can write $\partial_{\theta_i}{e^{-H}}$ as
\begin{align}
&\partial_{\theta_i} e^{-H} = \nonumber \\
&e^{-H}\left(-H_i-\frac{1}{2}\left[ H, H_i \right] - \frac{1}{6}\left[ H,\left[ H, H_i \right]\right] + \cdots \right).
\end{align}
We provide the details of the derivation in Appendix~\ref{app:proof-exp-mat-der}. Then, we rewrite $\Tr[\Lambda_v \partial_{\theta_i}{e^{-H}}]$ by inserting the matrix exponential derivative. Exploiting the linearity of the trace, we obtain the following expression:
\begin{align}
\label{tr-positive_ph}
\Tr[\Lambda_v \partial_{\theta_i}{e^{-H}}] = &- \Tr[\Lambda_v e^{-H}H_i] \nonumber \\
&- \frac{1}{2}\Tr[\Lambda_v e^{-H}\left[ H, H_i \right]] \nonumber \\
&- \frac{1}{6}\Tr[\Lambda_v e^{-H}\left[ H,\left[ H, H_i \right]\right]] \nonumber \\
&+ \cdots 
\end{align}

Next, we observe that if $[\Lambda_v, H]=0$ (while noting that $[H_i,H] \neq 0$ still holds), the second term in Eq.~\eqref{tr-positive_ph} simplifies to:
\begin{align}
\label{tr-simplification}
\Tr[\Lambda_v e^{-H}\left[ H, H_i \right]] &= \Tr[\Lambda_v e^{-H}H H_i] - \Tr[\Lambda_v e^{-H} H_i H]  \nonumber \\
&= \Tr[H \Lambda_v e^{-H} H_i] - \Tr[\Lambda_v e^{-H} H_i H] \nonumber \\
&= \Tr[\Lambda_v e^{-H} H_i H] - \Tr[\Lambda_v e^{-H} H_i H] \nonumber \\
&= 0,
\end{align}
where we leverage the commutation rules and the cyclic property of the trace. This can be repeated for the rest of the terms in Eq.~\eqref{tr-positive_ph} that contain commutators to observe that only $\Tr[\Lambda_v e^{-H}H_i]$ results in a non-zero value. Similarly, notice that $\Tr[\partial_{\theta_i}{e^{-H}}] = -\Tr[{e^{-H}} H_i]$ for all models, independent of the Hamiltonian definition. Then, the computation of gradients in Eq.~\eqref{eq:qrbm-grad} can be simplified, although the Hamiltonian still contains non-commuting terms.

Recall that we set $[\Lambda_v, H]=0$ to obtain simplified gradients for non-commuting Hamiltonians. Let us rewrite $\Lambda_{v}$ from Eq.~\eqref{eq:projector-comp-basis} in the Pauli basis as
\begin{equation}
    \Lambda_{v} = \left(\,\bigotimes_{i=1}^n \frac{1}{2} (I + (-1)^{v_i}Z ) \right) \otimes I^{\otimes m}.
\end{equation}
Then, the costly gradient computation of QRBMs can be avoided by choosing a non-commuting Hamiltonian that commutes with $\sigma_i^Z$ only over the visible units. Specifically, this condition can be satisfied by defining $\Hv$ with $\Wv=\{Z\}$ (see Definition~\ref{def-qrbm}), while the remaining terms of the Hamiltonian can be defined with an arbitrary Pauli operator set. We define such a model as a semi-quantum restricted Boltzmann machine (sqRBM).

\subsection{Semi-quantum Boltzmann machines}

\begin{definition}[Semi-quantum RBM]
\label{def-sqrbm}
A semi-quantum restricted Boltzmann machine with $n$ visible and $m$ hidden units, denoted $\sqRBM_{n,m}$, is described by a parameterized Hamiltonian according to Eq.~\eqref{eq:hamiltonian}. The three terms of the Hamilton are defined as follows:
\begin{align}
&\Hv = \sum_{i=1}^{n} a_{i}^{Z} \sigma_i^{Z}, \quad \Hh = \sum_{P \in \Wh}^{|\Wh|} \sum_{j=1}^{m} b_{j}^{P} \sigma_{n+j}^{P},  \nonumber \\
&\Hint = \sum_{P \in \Wh}^{|\Wh|} \sum_{i=1}^{n} \sum_{j=1}^{m} w_{i, j}^{Z,P} \sigma_{i}^{Z} \otimes \sigma_{n+j}^{k},
\end{align}
where $\boldsymbol{a} \in \mathbb{R}^{n}$, $\boldsymbol{b} \in \mathbb{R}^{|\Wh| \cdot m}$ and $\boldsymbol{w} \in \mathbb{R}^{|\Wh| \cdot n \cdot m}$ are the parameter vectors of the model. $\Wh$ is a non-commuting set of one-qubit Pauli operators that non-trivially act only on the hidden units.
\end{definition}

Notice that if one chooses $\Wh=\{X\}$, $\Wh=\{Y\}$ or $\Wh=\{Z\}$, $H$ is a commuting Hamiltonian, and these models are equivalent to an RBM. There are three possible non-commuting choices for an sqRBM ($\Wh = \{X, Z\}$, $\Wh = \{Y, Z\}$ or $\Wh = \{X, Y, Z\}$). Pauli sets $\Wh = \{X, Z\}$ and  $\Wh = \{Y, Z\}$ result in two equivalent models that differ by a change of basis. We provide various models with their number of parameters in Table~\ref{tab:model-classif}.

\begin{table}[!h]
    \centering
    \caption{\textbf{Classification of various models.}}
    \begin{tabularx}{8cm}{XXX}
        Model & $\Wh$  & Parameters  \\
        \hline
        \hline
        $\RBM_{n,m}$ & $\{ Z\}$ & $n+m(n+1)$ \\
        $\sqRBM_{n,m}$ & $\{X, Z\}$ & $n+2m(n+1)$ \\
        $\sqRBM_{n,m}$ & $\{Y, Z\}$ & $n+2m(n+1)$ \\
        $\sqRBM_{n,m}$ & $\{X, Y, Z\}$ & $n+3m(n+1)$ \\
        \hline
    \end{tabularx}
    \label{tab:model-classif}
\end{table}

One can leverage the similarity of sqRBMs to RBMs in order to obtain a closed-form expression for the output probabilities as well as their gradients. For an sqRBM, contributions of $X$, $Y$ and $Z$ terms are equivalent. Let us denote the state\footnote{Not to be confused with a quantum state.} of the hidden units for a given visible unit configuration $v$ with $\phi_j^{P}(v)$ such that
\begin{equation}
    \phi_j^{P}(v) = b_{j}^{k} + \sum_{i=1}^{n} (-1)^{v_i} w_{i,j}^{Z,P},
    \label{eq:hidden_state}
\end{equation}
where $P \in \{X, Y, Z\}$. Then, we define the following vector that combines the states of all possible three Pauli operators:
\begin{equation}
    \Phi_j(v) = \begin{bmatrix}
        \phi_j^{X}(v) & \phi_j^{Y}(v) & \phi_j^{Z}(v)
    \end{bmatrix}.
\label{eq:hidden-state-vector}
\end{equation}
This leads to Proposition~\ref{prop:sqrbm-probs}, which describes the output probabilities of sqRBMs.

\begin{proposition}[Output probabilities of sqRBM]\label{prop:sqrbm-probs} The unnormalized output probabilities of an $\sqRBM_{n,m}$ are as follows:
\begin{equation}
\tilde{p}_v = \left( \prod_{i=1}^{n} e^{-(-1)^{v_i}a_{i}^{Z}} \right) \left(  \prod_{j=1}^{m} \cosh(||\Phi_j(v)||_2) \right),
\label{eq:sqrbm-probs}
\end{equation}

where $v \in \{0, 1\}^n$ and $\Phi_j(v)$ is the vector of hidden states as described in Eq.~\eqref{eq:hidden-state-vector}. Then, the normalized output probabilities are given as
\begin{equation}
    p_v = \tilde{p}_v / \sum_v \tilde{p}_v.
\end{equation}
\end{proposition}

The proof is provided in Appendix~\ref{proof:prop:sqrbm-probs}. Next, we provide the gradients of sqRBMs following the result of Proposition~\ref{prop:sqrbm-probs}.

\begin{proposition}[Gradients of sqRBM]\label{prop:sqrbm-grads} An $\sqRBM_{n,m}$ with the output probability distribution $p$ can be trained to minimize the negative log-likelihood with respect to the target probability distribution $q$ using the following gradient rule:

Gradients of the parameters for the field terms acting on the visible units are given as 
\begin{equation}
\partial_{a_i}{\mathcal{L}} = \sum_v q_v \left( (-1)^{v_i} - \sum_v (-1)^{v_i} p_v \right).   
\end{equation}

Gradients of the parameters for the field terms acting on the hidden units are given as 
\begin{align}
\partial_{b_j^{P}}{\mathcal{L}} = &-\sum_v q_v \frac{\phi_j^P(v) }{||\Phi_j(v)||_2} \tanh \left( ||\Phi_j(v)||_2 \right) \nonumber \\
&+\sum_v q_v \sum_v \frac{\phi_j^P(v)}{||\Phi_j(v)||_2} \tanh \left( ||\Phi_j(v)||_2 \right)  p_v.    
\end{align}

Gradients of the parameters for the interaction terms acting on both visible and hidden units are given as 
\begin{align}
\partial_{w_{i,j}^{Z,P}}{\mathcal{L}} &= -\sum_v q_v  (-1)^{v_i}\frac{\phi_j^P(v) }{||\Phi_j(v)||_2} \tanh \left( ||\Phi_j(v)||_2 \right) \nonumber \\
&+ \sum_v q_v \sum_v (-1)^{v_i}\frac{\phi_j^P(v) }{||\Phi_j(v)||_2} \tanh \left( ||\Phi_j(v)||_2 \right)  p_v.
\end{align}
\end{proposition}

The proof is provided in Appendix~\ref{proof:prop:sqrbm-grads}. Notice that both the gradients and output probabilities for RBM and sqRBM contain summations over all visible configurations, which grow exponentially in input size. This causes both of these models to be intractable for large input sizes on classical computers.

So far, we have focused on BM variants with restricted connectivity. Incorporating lateral connections within visible and hidden layers increases the computational cost of the training procedure~\cite{salakhutdinov_efficient_2012}. Although additional connectivity enhances the expressive power of BMs, computational complexity limits their practical applicability. However, training fully-connected QBMs without a significant computational overhead may be discovered within the quantum computing framework. Motivated by this opportunity, we introduce a fully-connected sqBM.  

\begin{definition}[Semi-quantum BM]
\label{def-sq-bm}
A semi-quantum Boltzmann machine (sqBM) with $n$ visible and $m$ hidden units is denoted $\sqBM_{n,m}$. An sqBM is a generalization of an sqRBM with lateral connections in the visible and hidden layers, which is described by a parameterized Hamiltonian according to Eq.~\eqref{eq:hamiltonian}. The three terms of the Hamilton are defined as follows:
\begin{align}
&\Hv = \sum_{i=1}^{n} \theta_{i}^{Z} \sigma_i^{Z} +  \sum_{i=1}^{n-1} \sum_{j=i+1}^{n} \theta_{i, j}^{Z,Z} \sigma_i^{Z} \otimes \sigma_j^{Z},  \nonumber \\
&\Hh = \sum_{k \in \Wh}^{|\Wh|} \left(\, \sum_{i=n+1}^{n+m} \theta_{i}^{k} \sigma_i^{k} \right. \nonumber \\  
& \qquad \qquad \qquad \left. + \sum_{l \in \Wh}^{|\Wh|} \sum_{i=n+1}^{n+m-1} \sum_{j=i+1}^{n+m} \theta_{i, j}^{k,l} \sigma_i^{k} \otimes \sigma_j^{l} \right),  \nonumber \\
&\Hint = \sum_{k \in \Wh}^{|\Wh|} \sum_{i=1}^{n} \sum_{j=n+1}^{n+m} \theta_{i, j}^{Z,k} \sigma_{i}^{Z} \otimes \sigma_{j}^{k},
\end{align}
where $\Wh$ is a non-commuting set of one-qubit Pauli operators that non-trivially act only on the hidden units.
\end{definition}

Ultimately, the results from Eq.~\eqref{tr-positive_ph} and Eq.~\eqref{tr-simplification} lead to a simplified expression for the gradients of a generic sqBM. Importantly, satisfying the condition $[\Lambda_v, H]=0$ does not impose restrictions on the connectivity between units. Then, one obtains the following closed-form expression for the gradients that is significantly cheaper to compute than Eq.~\eqref{eq:qrbm-grad}.

\begin{proposition}[Gradients of sqBM]\label{prop:sqbm-grads} An $\sqBM_{n,m}$ can be trained to minimize the negative log-likelihood with respect to the target probability distribution $q$ using the following gradient rule:
\begin{equation}
    \partial_{\theta_i}{\mathcal{L}} = -\sum_v q_v \left( \underbrace{-\frac{\Tr[\Lambda_v e^{-H}H_i]}{\Tr[\Lambda_v e^{-H}]}}_{\mathrm{\text{positive phase}}} + \underbrace{\frac{\Tr[e^{-H}H_i]}{\Tr[e^{-H}]}}_{\mathrm{\text{negative phase}}} \right),
    \label{eq:sqRBM-grad}
\end{equation}
where $\theta_i$ is any real-valued parameter of the model, when the Hamiltonian terms are grouped such that $H = \sum_i \theta_i H_i$.
\end{proposition}

The proof is provided in Appendix~\ref{proof:prop:sqbm-grads}. Notice that Proposition~\ref{prop:sqbm-grads} provides the recipe to compute the gradients of either an sqRBM or the more general sqBM on a quantum computer. In this work, we study the relationship between RBMs and sqRBMs; therefore, in the remainder of the work, we will not consider sqBMs and leave their study as future work.

\subsection{Expressive equivalence}

Observe that output probabilities provided in Proposition~\ref{prop:sqrbm-probs} are similar whether one considers an RBM or an sqRBM. Naturally, this brings up several questions, such as: Is there a connection between these two models? Can we establish a mapping between them? Do the gradients of sqRBMs behave in the same way as those of RBMs? 

To help us answer some of these questions, let us consider an $\sqRBM_{n,1}$ with a single hidden unit. Let us set the parameters of the field terms that act on the visible units to zero only for simplicity. Then, following Eq.~\eqref{eq:sqrbm-probs}, the unnormalized output probability of this model can be written as 
\begin{equation}
\tilde{p}_v = \cosh(\sqrt{\phi_j^X(v)^2+\phi_j^Y(v)^2+\phi_j^Z(v)^2}).
\label{eq:sqrbmnn1-probs}
\end{equation}
Then, using Corollary~\ref{corollary:cosh} from Appendix~\ref{app:defs}, for small values of all parameters, the output probabilities can be written as
\begin{equation}
\tilde{p}_v \simeq \cosh(\phi_j^X(v)) \cosh(\phi_j^Y(v))\cosh(\phi_j^Z(v)),
\label{eq:sqrbmnn1-probs2}
\end{equation}
which is precisely the form that corresponds to the output probabilities of an $\RBM_{n,3}$ with three hidden units (see Eq.~\eqref{eq:sqrbm-probs}). In fact, these models also have the same number of parameters. Therefore, for small values of the parameters, there exists one-to-one mapping parameters, which give the same output probability distribution. For larger values of the parameters, the mapping of the parameters becomes non-trivial, but the expressivity of the models remains the same. We express this result more formally in Theorem~\ref{theorem:equivalance}.

\begin{theorem}[Equivalence of sqRBM hidden units to RBM multiple hidden units.] \label{theorem:units}
The single hidden unit of an $\sqRBM_{n,1}$ with $n$ visible units and operator set $\Wh$ is equivalent to $|\Wh|$ hidden units in an $\RBM_{n,|\Wh|}$. There exists a bijective mapping between the parameter spaces for small parameter values, while for larger values, the mapping becomes non-trivial and structurally dependent on $\Wh$.
\end{theorem}

The proof follows Eq.~\eqref{eq:sqrbmnn1-probs2}, and it is provided in Appendix~\ref{proof:theorem:equivalance}. Following Theorem~\ref{theorem:units}, one can generalize the statement to multiple hidden units for sqRBMs and obtain Theorem~\ref{theorem:equivalance}.

\begin{theorem}[Equivalence of sqRBM to RBM.] \label{theorem:equivalance}
The hidden units of an $\sqRBM_{n,m}$ with $n$ visible, $m$ hidden units and operator set $\Wh$ are equivalent to $|\Wh| \cdot m$  hidden units in an $\RBM_{n,|\Wh|\cdot m}$. There exists a bijective mapping between the parameter spaces for small parameter values, while for larger values, the mapping becomes non-trivial and structurally dependent on $\Wh$.
\end{theorem}
\begin{proof}
    The theorem follows Theorem~\ref{theorem:units}. For each hidden unit of an $\sqRBM$, one obtains $|\Wh|$-many $\RBM$ hidden units. According to Proposition~\ref{prop:sqrbm-probs}, hidden units are independently added as products of each other. Then, hidden units of an $\sqRBM_{n,m}$ with the operator pool $\Wh$ are equivalent to $|\Wh| \cdot m$ hidden units with the same arguments for the mapping of the parameters as these models have exactly the same number of parameters.
\end{proof}

The consequence of Theorem~\ref{theorem:equivalance} is that an $\sqRBM_{n,m}$ with the operator pool $\Wh$ can solve the tasks that an $\sqRBM_{n,|\Wh| \cdot m}$ can solve equally well with the same number of parameters. In other words, these models have the same expressivity. In the next section, we support Theorem~\ref{theorem:equivalance} with numerical simulations.

\subsection{Numerical results}\label{sec:numerics}

\begin{figure*}[!t]
    \centering
    \includegraphics[width=\linewidth]{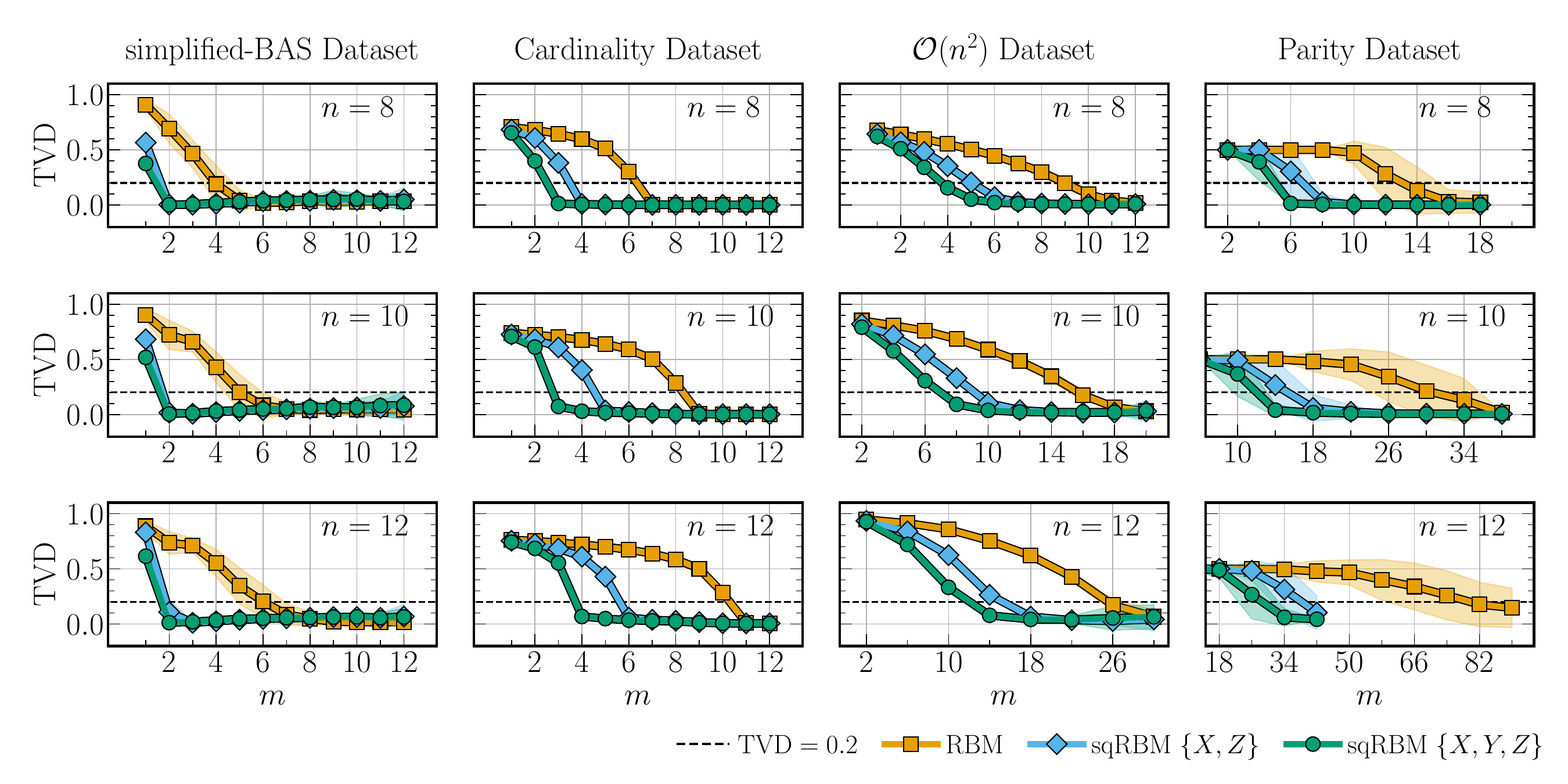}
    \caption{\textbf{Training results.} We train three models ($\RBM$, $\sqRBM \{X, Z\}$ and $\sqRBM \{X, Y, Z\}$) over four datasets with three different input sizes ($n \in \{8,10,12\}$) and various number of hidden units in the range $m \in [1,90]$. We report the total variation distance (TVD) measured after training all models 100 times with different initial parameters. The solid lines report the average, while the shades indicate the standard deviation. Each column reports results for a different dataset, ordered in increasing difficulty from left to right. The target probability distribution for $\mathcal{O}(n^2)$ dataset is varied for each run, using the same 100 seed for all models. The same target probability distribution is used for the other datasets in all runs. \looseness-1}
    \label{fig:training-results}
\end{figure*}

In this section we provide numerical results to support our theoretical findings. The code and the generated data to reproduce the results can be accessed online~\cite{code, data}. We perform exact numerical simulations to train various RBM and sqRBM models. We use four distinct datasets, which are defined as follows: 
\begin{itemize}
    \setlength{\itemsep}{0pt} 
    \setlength{\parindent}{0pt} 
    \setlength{\leftskip}{0pt} 
    \renewcommand{\labelitemi}{} 
    \item \textbf{simplified-BAS dataset:} $n$-bit uniform probability distribution over the bitstrings that correspond to vertical or horizontal lines on a $2 \times n/2$ grid.
    \item \textbf{$\boldsymbol{\mathcal{O}}{\mathbf{(\textit{n}^2)}}$ dataset:} $n$-bit uniform probability distribution over randomly chosen $n^2$ bitstrings.
    \item \textbf{Cardinality dataset:} $n$-bit uniform probability distribution over the bitstrings that have $n/2$ cardinality.
    \item \textbf{Parity dataset:} $n$-bit uniform probability distribution over the bitstrings that have even parity.
\end{itemize}

We train all models 100 times with parameters randomly initialized from a uniform distribution between $[-1,1]$ using the \textsc{AMSGrad} optimizer~\cite{amsgrad} with the hyperparameters $\{\mathrm{lr}=0.1, \beta_1=0.9, \beta_2=0.999 \}$ and $\KL$ as the loss function. We have not performed hyperparameter optimization as all models converge within a reasonable number of iterations. We emphasize that all results can be improved with dedicated hyperparameter optimization. Our goal in this work is not to find the best performing model but to treat all models the same way.

We employ an $\RBM$ and two quantum models: $\sqRBM \{X, Z\}$ and $\sqRBM \{X, Y, Z\}$. These quantum models are defined by two sets of operators for the hidden units, $\Wh=\{X,Z\}$ and $\Wh=\{X,Y,Z\}$, respectively. Since we minimize $\KL$, we report values of another metric, namely total variation distance ($\TVD(q || p) = \frac{1}{2}||p - q||_1$) for both RBM and sqRBM models across the described datasets in Figure~\ref{fig:training-results}. As a reference point for practical purposes, we report the $\TVD=0.2$ line. It can be seen that all models can go below the $\TVD=0.2$ line given that they have sufficient number of hidden units. 

In general, the number of hidden units required to achieve a good approximation depends on the dataset. More specifically, the support of a probability distribution is a good measure of difficulty. In Figure~\ref{fig:training-results}, the datasets are ordered in increasing difficulty from left to right. There, we observe that one needs a larger value of $m$ for all models as the support of the dataset increases. In practice, it is considered that a dataset has $\mathrm{supp}(p) \in \mathcal{O}(\mathrm{poly}(n))$, therefore, it is expected that $m \in \mathcal{O}(\mathrm{poly}(n))$ suffices to learn a distribution with good precision. 

Our findings across the considered datasets indicate that $\sqRBM \{X, Y, Z\}$ requires fewer hidden units to reach the reference point compared to $\sqRBM \{X, Z\}$. Similarly, $\sqRBM \{X, Z\}$ requires fewer hidden units to reach the reference point compared to $\RBM$. However, all models are able to reach the same low values of $\TVD$, provided they have a sufficient number of hidden units. Theorem~\ref{theorem:equivalance} predicts $\sqRBM_{n, m} \{X, Z\}$ to perform equally as well as $\RBM_{n, 2m}$ (with $2m$ hidden units in the RBM) and similarly for $\sqRBM_{n, m} \{X, Y, Z\}$ to perform equally as well as $\RBM_{n, 3m}$. We emphasize that the models have the same number of parameters for these settings. We numerically verify Theorem~\ref{theorem:equivalance} by computing the ratios of the number of hidden units that are sufficient for each model to achieve $\TVD < 0.2$ in Figure~\ref{fig:threshold-results}. As expected, the minimum number of hidden units $m$, that is sufficient to go below the threshold, varies with dataset and input size. Overall, results align with the prediction of Theorem~\ref{theorem:equivalance} across four datasets and varying input sizes.

\begin{figure*}[!t]
    \centering
    \includegraphics[width=\linewidth]{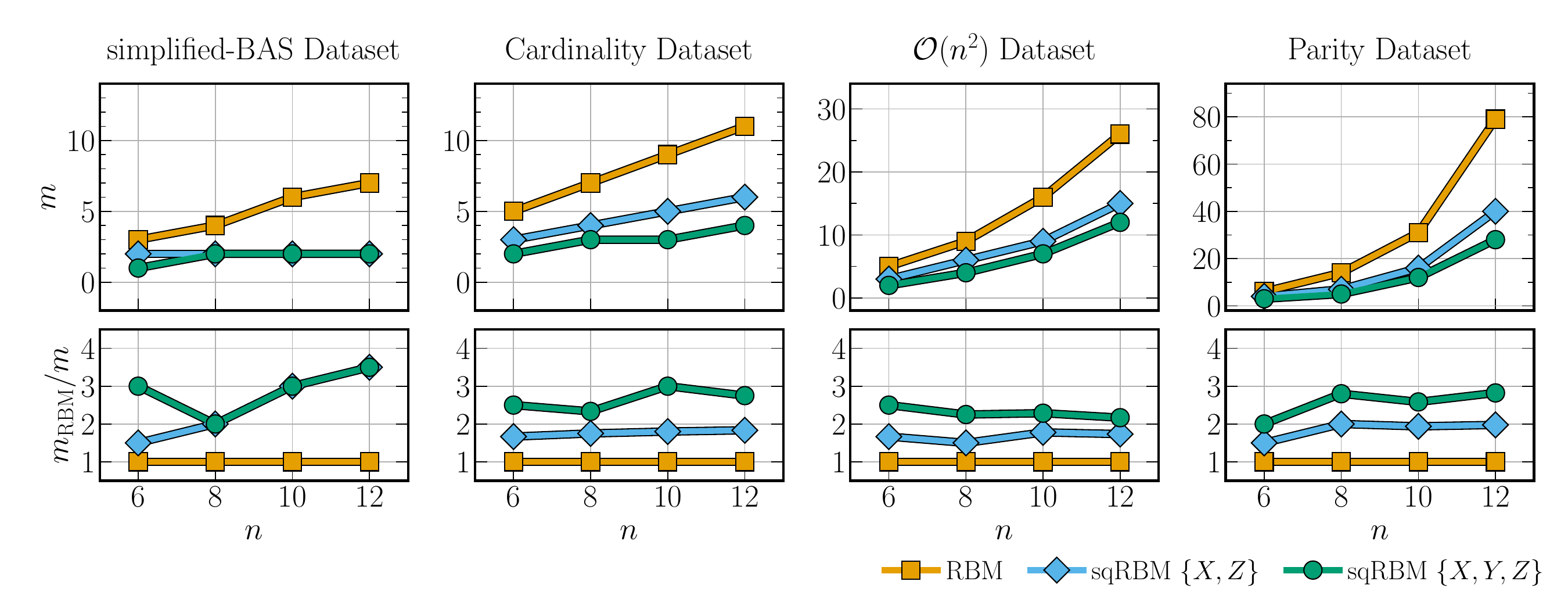}
    \caption{\textbf{Minimum number of hidden units required to learn target probability distributions on average.} We report the minimum number of hidden units $m$ required to achieve $\mathrm{TVD} < 0.2$ on average, over four datasets for various input sizes ($n \in \{6,8,10,12\}$) using three models ($\RBM$, $\sqRBM \{X, Z\}$ and $\sqRBM \{X, Y, Z\}$) as in Figure~\ref{fig:training-results}. In the bottom panel we provide the ratio of $m_\mathrm{RBM}$ to $m$ of the other models. The difficulty of each dataset results in a different scaling behavior. Recall that $\RBM$ has the same number of parameters and expressivity according to Theorem~\ref{theorem:equivalance} as $\sqRBM \{X, Z\}$ for the ratio $m_\mathrm{RBM} = 2m$ and similarly $\sqRBM \{X, Y, Z\}$ for $m_\mathrm{RBM} = 3m$. \looseness-1}
    \label{fig:threshold-results}
\end{figure*}

\section{Discussion} \label{sec:discussion}

In this work, we introduce semi-quantum restricted Boltzmann machines (sqRBM), a subclass of quantum Boltzmann machines (QBM). The sqRBM Hamiltonian acts with arbitrary non-commuting Pauli operators on hidden units, whereas operators for visible units are restricted to a commuting set. The structure provided by the Hamiltonian allows us to circumvent the expensive computation of the positive phase in Proposition~\ref{prop:qrbm-grads} and enables us to derive the analytical output probabilities and gradients in Proposition~\ref{prop:sqrbm-probs} and Proposition~\ref{prop:sqrbm-grads}.

In the field of quantum machine learning, efficient evaluation of gradients remains a major challenge. Expressive models based on parameterized quantum circuits exhibit trainability issues, often linked to barren plateaus of the loss landscape~\cite{mcclean_barren_2018}. Barren plateaus lead to an exponential increase in the number of samples required to evaluate gradients as the system size grows. This significantly limits the scalability of circuit-based quantum machine learning models, making it difficult to demonstrate their practical viability.

QBMs offer a promising alternative to circuit-based approaches. However, when the entanglement entropy between visible and hidden units obeys the volume law, generic QRBM models are susceptible to barren plateaus~\cite{ent-bp}. Importantly, sqRBMs do not have entanglement between visible and hidden units, therefore, this problem is naturally mitigated, and the gradients can be estimated with polynomially many samples from the Gibbs state. Moreover, the similarity between RBM and sqRBM gradients also suggests the absence of vanishing gradients. For completeness, we provide numerical results that confirm this conjecture in Appendix~\ref{app:numerics-grads}.

Expressivity of QBMs is directly related to their Hamiltonian. A popular choice in the literature is to employ the transverse field Ising model (TFIM)~\cite{amin_quantum_2018} due to its similarity to the classical Ising model. However, QBMs based on TFIM have been reported to exhibit only marginal improvement in learning capacity compared to BMs, while generic Hamiltonians demonstrate significant improvement~\cite{tuysuz2024learninggeneratehighdimensionaldistributions}. Our results help explain the poor performance of such models and provide a strategy to build more expressive ones. 

Let us consider a QRBM based on the TFIM Hamiltonian. Such a model is defined by adding transverse $X$ fields to visible and hidden units of the RBM Hamiltonian. In our notation, this can be expressed as $\Wv=\Wh=\{X,Z\}$ and $\Wint=\{ZZ\}$. Recall that the state of the $j$-th hidden unit is described with $||\Phi_j(v)||_2$ (see Eq.~\eqref{eq:sqrbm-probs}). Then, for a TFIM Hamiltonian, one obtains {\tiny $\sqrt{\phi_j^Z(v)^2 + (b_j^X)^2}$}, where $b_j^X$ is the parameter of the transverse $X$ field of the TFIM Hamiltonian on the $j$-th hidden unit. Then, it follows that the contribution of the transverse field is independent of the visible unit configuration $v$ and can not significantly change the model's expressivity. Hence, the output probability distributions of such QBMs closely resemble those of classical models. One could improve the expressivity of such QBM models by including $ZX$ terms in the Hamiltonian to have $\Wv=\Wh=\{X,Z\}$ and $\Wint=\{ZZ,ZX\}$. Note that such a model is a QRBM and not an sqRBM. Therefore, its gradients are not efficiently computable.

The practical feasibility of machine learning models depends both on their learning capacity and the computational resources required for training and sampling. The computational cost of training an $\RBM_{n,m}$ using exact analytical gradients, as well as sampling from a trained model, scales as $\mathcal{O}(\exp(n+m))$ on a classical computer~\cite{kappen_efficient}. However, training costs can be significantly reduced by leveraging classical techniques such as contrastive divergence (CD)~\cite{bm_cd}. Analogously, given the similarity in the output probability expressions of sqRBM and RBM, a CD-based training algorithm for sqRBM may be developed as well. We leave the study of CD algorithms to train sqRBMs as future work.


While the training cost of RBMs, and potentially sqRBMs, can be mitigated, sampling from both models remains a computationally demanding task. Alternatives to facilitate approximate or exact Gibbs state preparation, including quantum algorithms, may be employed to further decrease the computational cost~\cite{mean_field_kappen, salakhutdinov_efficient_2012, chen2023quantumthermalstatepreparation}. There are multiple proposals of quantum algorithms that offer Gibbs state preparation with $\mathcal{O}(\mathrm{poly}(n,m))$ cost~\cite{chen2023quantumthermalstatepreparation} Therefore, access to fault-tolerant quantum computers could provide a sampling advantage for both RBMs and sqRBMs.

Compared to RBMs, sqRBMs require fewer qubits on a quantum computer to perform the same task. Therefore, sqRBM presents a promising direction for quantum machine learning, demonstrating that the benefits of quantum models can originate from reduced resource requirements rather than speedups. By requiring fewer quantum resources, sqRBMs could improve feasibility for early quantum machine learning applications, particularly in generative modeling and feature selection.

One of the goals of quantum machine learning is to process both classical and \textit{quantum} data efficiently. Although sqRBMs are quantum models, they only support classical data as input because their Hamiltonian is commuting within the subspace of visible units. Notably, there are proposals in the literature for QBMs that are suitable for quantum data~\cite{wiebe_generative_2019}.  

In future work, adding lateral connections between hidden units could be utilized to further enhance the expressivity of sqRBMs. To this end, we introduce the definition~\ref{def-sq-bm} of sqBM that has additional connectivity within visible and hidden units. Note that the gradient expression for sqBM is the same as for sqRBM. However, increased connectivity within hidden units may result in entanglement-induced barren plateaus~\cite{ent-bp}. Another promising extension to improve the expressivity of sqRBMs is to incorporate additional hidden layers, similar to deep Boltzmann machines~\cite{salakhutdinov_deep_2009}.

\begin{acknowledgments}
C.T. is supported in part by the Helmholtz Association -``Innopool Project Variational Quantum Computer Simulations (VQCS)''. This work is supported with funds from the Ministry of Science, Research and Culture of the State of Brandenburg within the Centre for Quantum Technologies and Applications (CQTA). This work is funded within the framework of QUEST by the European Union’s Horizon Europe Framework Programme (HORIZON) under the ERA Chair scheme with grant agreement No.\ 101087126. Parts of this research have been funded by the Federal Ministry of Education and Research of Germany and the state of North-Rhine Westphalia as part of the Lamarr-Institute for Machine Learning and Artificial Intelligence. M.G. is supported by CERN through the CERN Quantum Technology Initiative.
\begin{figure}[!h]
    \centering
    \includegraphics[width=0.1\textwidth]{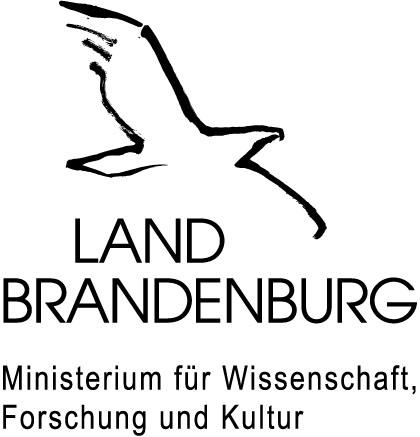}
\end{figure}
\end{acknowledgments}

\bibliography{main}

\onecolumngrid
\appendix
\section*{Appendix}

\makeatletter
\renewcommand\thesubsection{\thesection\arabic{subsection}}
\renewcommand\p@subsection{}
\makeatother

\section{Nomenclature}

\begin{table}[!h]
    \centering
    \renewcommand{\arraystretch}{1.5} 
    \begin{tabularx}{\linewidth}{|C{2cm}|X|}
    \hline
    $a_i^k$ & Real-valued parameter of the field term acting on the $i$-th visible unit with Pauli-$k$ or Pauli-$Z$ if $k$ is not specified.\\
    \hline
    $\boldsymbol{a}$ & Vector of all real-valued parameters of the field terms acting on visible units. \\
    \hline 
    $b_j^k$ & Parameter of the field term acting on the $j$-th hidden unit with Pauli-$k$ \\
    \hline
    $\boldsymbol{b}$ & Vector of all real-valued parameters of the field terms acting on hidden units. \\
    \hline
    $w_{i,j}^{k,l}$ & Parameter of the interaction term acting on the $i$-th visible and $j$-th hidden unit with Pauli $k \otimes l$. \\
    \hline 
    $\boldsymbol{w}$ & Vector of all real-valued parameters of the interaction terms acting on visible and hidden units.\\
    \hline 
    $\theta_i$ & The real-valued parameter of the Hamiltonian term $H_i$, when the Hamiltonian terms are grouped as $H = \sum_i \theta_i H_i$.\\
    \hline 
    $\boldsymbol{\theta}$ & The vector of all real-valued parameters of the Hamiltonian, when the Hamiltonian terms are grouped as $H = \sum_i \theta_i H_i$.\\
    \hline
    $\Hv$ & Hamiltonian that acts non-trivially on the visible units. \\
    \hline
    $\Hh$ & Hamiltonian that acts non-trivially on the hidden units. \\
    \hline
    $\Hint$ & Interaction Hamiltonian that acts non-trivially on visible and hidden units at the same time. \\
    \hline
    $\mathcal{W}$ & Set of Pauli operators. \\
    \hline
    $\KL$ & Kullback-Leibler divergence. \\
    \hline
    $\mathcal{L}$ & Loss function. \\
    \hline
    $\mathrm{TVD}$ & Total variation distance. \\
    \hline
    $\phi_j^P(v)$ & The state of the $j$-th hidden unit for operator $P$ for the given visible configuration $v$. \\
    \hline
    $\Phi_j(v)$ & The state vector of the $j$-th hidden unit for all operators for a given visible configuration $v$. \\
    \hline
    \end{tabularx}
    \label{tab:nomenclature}
\end{table}

\newpage
\section{Proofs}

\subsection{Useful definitions}\label{app:defs}

In this work we refer to Pauli operators with two different notations to improve readability. Below, we provide the matrices that correspond to these operators for completeness.

\begin{equation}
    I = \begin{bmatrix}
        1 & 0 \\
        0 & 1
    \end{bmatrix}, \quad 
    X=\sigma^X = \begin{bmatrix}
        0 & 1 \\
        1 & 0
    \end{bmatrix}, \quad
    Y=\sigma^Y = \begin{bmatrix}
        0 & -i \\
        i & 0
    \end{bmatrix}, \quad
    Z=\sigma^Z = \begin{bmatrix}
        1 & 0 \\
        0 & -1
    \end{bmatrix}.
\end{equation}

\begin{lemma}\label{lemma:cosh} Let $a$, $b$ and $c$ be real-valued scalars. For sufficiently small values of $b$ and $c$ such that $(bc)^2 \ll 6$, the equation $\cosh (a) = \cosh(b)\cosh(c)$, admits the approximate solution $a^2 \approx b^2 + c^2$.
\end{lemma}
\begin{proof} The Taylor series expansion of $\cosh(x)$ for real-valued values of $x$ is given by
\begin{equation}
    \cosh(x) = 1 + \frac{x^2}{2!} + \frac{x^4}{4!} + \mathcal{O}(x^6)
\end{equation}
Substituting this expansion into the product $\cosh(b) \cosh(c)$, we obtain
\begin{align}
\cosh(b) \cosh(c) &= \left( 1 + \frac{b^2}{2!} + \frac{b^4}{4!} + \mathcal{O}(b^6) \right) \left( 1 + \frac{c^2}{2!} + \frac{c^4}{4!} + \mathcal{O}(c^6) \right) \nonumber \\
&=  1 + \frac{b^2}{2!}+ \frac{c^2}{2!} + \frac{b^4}{4!} + \frac{c^4}{4!} + \frac{b^2 c^2}{2! 2!} + \cdots
\end{align}
For small $b$ and $c$, such that  $(bc)^2 \ll 6$, we can neglect the additional $(bc)^2$  and higher-order terms, leading to the approximation
\begin{equation}
\cosh(b) \cosh(c)  \approx 1 + \frac{(\sqrt{b^2 + c^2})^2}{2!} + \frac{(\sqrt{b^2 + c^2})^4}{4!} + \cdots.
\end{equation}
Comparing this with the Taylor series of $\cosh(x)$, we recognize that
\begin{equation}
\cosh(b) \cosh(c)  \approx \cosh(\sqrt{b^2 + c^2}).
\end{equation}
Thus, for small $a$, $b$ and $c$, the equation $\cosh(a) = \cosh(b) \cosh(c)$, admits the approximate solution $a^2 \approx b^2 + c^2$.
\end{proof}

\begin{corollary}\label{corollary:cosh} Let $a$, $b$, $c$ and $d$ be real-valued scalars. For sufficiently small values of $b$, $c$ and $d$, the equation $\cosh (a) = \cosh(b)\cosh(c)\cosh(d)$, admits the approximate solution $a^2 \approx b^2 + c^2 + d^2$.
\end{corollary}
\begin{proof}
    The corollary follows Lemma~\ref{lemma:cosh}. Recall the approximation $\cosh(b)\cosh(c) \approx \cosh(\sqrt{b^2 + c^2})$ for sufficiently small values of $b$ and $c$. Then, one can approximate $\cosh(\sqrt{b^2 + c^2})\cosh(d)$ as $\cosh(\sqrt{b^2 + c^2 + d^2})$ without loss of generality. Hence, the equation $\cosh(a) = \cosh(b)\cosh(c)\cosh(d)$ admits the approximate solution $a^2$ = $b^2 + c^2 + d^2$ with the assumptions of Lemma~\ref{lemma:cosh}. 
\end{proof}

\subsection{Proof of Proposition~\ref{prop:qrbm-grads}}
\label{app:proof-prop1}
\begin{proof}
Let us insert Eq.~\eqref{eq:state_v} into Eq.~\eqref{eq:loss_function} to obtain the loss for a target probability distribution $q$ and model with Hamiltonian $H$ as follows:
\begin{equation}
    \mathcal{L} = -\sum_v  q_v \log {\frac {\text{Tr}[\Lambda_v e^{-H}]} {\text{Tr} [e^{-H}]} } .
\end{equation}
Then, the gradients of the loss function with respect to the parameter $\theta_i$ can be obtained using the chain rule such that
\begin{align}
    \partial_{\theta_i}{\mathcal{L}} &= -\sum_v q_v \frac{\Tr[e^{-H}]}{\Tr[\Lambda_v e^{-H}]} \left( \frac{\partial_{\theta_i}{\Tr[\Lambda_v e^{-H}] \Tr[e^{-H}]} - \Tr[\Lambda_v e^{-H}] \partial_{\theta_i}{\Tr[e^{-H}]} }{\Tr [e^{-H}]^2}\right) \\
    &= -\sum_v q_v \left( \frac{\Tr[\Lambda_v \partial_{\theta_i}{e^{-H}}]}{\Tr[\Lambda_v e^{-H_i}]} - \frac{\Tr[\partial_{\theta_i}{e^{-H}}]}{\Tr[e^{-H}]} \right).
\end{align}
\end{proof}

\subsection{Derivation of derivative of the matrix exponential}
\label{app:proof-exp-mat-der}

Recall that the matrix exponential $\exp (-H)$ can be written as the power series such that
\begin{equation}
    e^{-H} = \sum_{k=0}^{\infty} \frac{(-1)^k }{k!} H^k,
    \label{eq:power-series}
\end{equation}
where $H^0$ is defined as the identity. Then, 
let us recall the following:
\begin{equation}
    \frac{d}{dt}e^{X(t)} = e^{X(t)}\frac{1-e^{-\text{ad}_X}}{\text{ad}_X} \frac{dX(t)}{dt},
\end{equation}
and $\text{ad}_X$ is the adjoint representation given as $\text{ad}_X ( \cdot) = \left[ X, \cdot \right]$~\cite{rossmann_lie_2006}. Then, we also write the following power series:
\begin{equation}
\frac{1-e^{-\text{ad}_X}}{\text{ad}_X} = \sum_{k=0}^{\infty} \frac{(-1)^k}{(k+1)!} \left( \text{ad}_X \right)^k.
\end{equation}

Then, let us insert $X(t)=-H$ and $t=\theta_i$. Then, $dX(t)/dt = -\partial H/\partial \theta_i = -H_i$ and the complete expression reads
\begin{align}
\frac{\partial}{\partial \theta_i}e^{-H} &= e^{-H} \left( \mathbb{1} -\frac{1}{2} \text{ad}_{-H} +\frac{1}{6} (\text{ad}_{-H})^2 + \cdots \right) (-H_i) \nonumber \\
&= e^{-H}\left(-H_i-\frac{1}{2}\left[ H, H_i \right]-\frac{1}{6}\left[ H,\left[ H, H_i \right]\right] + \cdots \right)
\end{align}

\subsection{Proof of Proposition~\ref{prop:sqrbm-probs}}
\label{proof:prop:sqrbm-probs}
\begin{proof}
Let us consider an $\sqRBM_{1,1}$ with $\Wh= \{ X, Z\}$. Then the output probabilities are defined as
\begin{equation}
    p_v = \Tr[\Lambda_v \frac{e^{-H}}{\Tr[e^{-H}]}],
\end{equation}
where $H$ is the Hamiltonian of the model. Notice that one can define the unnormalized probabilities by taking the trace $\Tr[e^{-H}]$ outside, such that 
\begin{equation}
    \tilde{p}_v = \Tr[\Lambda_v e^{-H}].
\end{equation}

Then, it is possible to compute the unnormalized probabilities as follows:
\begin{equation}
\tilde{p}_0 = \Tr[\begin{bmatrix}
       1 & 0 & 0 & 0 \\
       0 & 1 & 0 & 0 \\
       0 & 0 & 0 & 0 \\
       0 & 0 & 0 & 0 \\
   \end{bmatrix} e^{-H} ] \quad \text{and} \quad 
\tilde{p}_1 = \Tr[\begin{bmatrix}
       0 & 0 & 0 & 0 \\
       0 & 0 & 0 & 0 \\
       0 & 0 & 1 & 0 \\
       0 & 0 & 0 & 1 \\
   \end{bmatrix} e^{-H} ].
\end{equation}

Next, let us express the Hamiltonian $H$ as a matrix:

\begin{equation}
    H = \begin{bmatrix}
        a_1 + b_1^Z + w_{1,1}^Z &  b_1^X + w_{1,1}^X & 0 & 0 \\
        b_1^X + w_{1,1}^X & + a_1 - b_1^Z - w_{1,1}^Z & 0 & 0 \\
        0 & 0 & -a_1 + b_1^Z - w_{1,1}^Z & b_1^X - w_{1,1}^X \\
        0 & 0 & b_1^X - w_{1,1}^X & -a_1 - b_1^Z + w_{1,1}^Z
    \end{bmatrix}.
\end{equation}

Then, since $H$ is not a diagonal matrix, one has to diagonalize the matrix $H$ in order to compute $e^{-H}$. This can be done with the help of tools such as $\textsc{Mathematica}$. Then, one obtains $\tilde{p}_0$ and $\tilde{p}_1$ as follows:

\begin{align}
\tilde{p}_0 &= e^{-a_1} 2 \cosh \left(\sqrt{ (b_1^Z + w_{1,1}^Z)^2 + (b_1^X + w_{1,1}^X)^2 }\right), \nonumber \\
\tilde{p}_1 &= e^{a_1} 2 \cosh \left(\sqrt{ (b_1^Z - w_{1,1}^Z)^2 + (b_1^X - w_{1,1}^X)^2 }\right).
\label{eq:probs-sqrbm1-xz}
\end{align}

Notice that all terms that have a Pauli-$Z$ acting on the visible unit get a sign of the visible configuration such that $(-1)^{v_i}$, where $\forall i, v_i \in \{0, 1\}$. Also, the hidden units that corresponds to $X$ and $Z$ form exactly the same expression such that their contribution $\phi_j^k(v)$ is identical and differ only by the parameters of the term. Then, one can generalize this for an $\sqRBM_{n,m}$ with $n$ visible and $m$ hidden units by defining the following:
\begin{equation}
    \phi_j^{k}(v) = b_{j}^{k} + \sum_{i=1}^{n} (-1)^{v_i} w_{i,j}^{Z,k},
\end{equation}
\begin{equation}
    \Phi_j(v) = \begin{bmatrix}
        \phi_j^{X}(v) & \phi_j^{Y}(v) & \phi_j^{Z}(v)
    \end{bmatrix}
\end{equation}
\begin{equation}
    \tilde{p}_v = \left( \prod_{i=1}^{n} e^{-(-1)^{v_i}a_{i}^{Z}} \right) \left(  \prod_{j=1}^{m} \cosh(||\Phi_j(v)||_2) \right),
\end{equation}
where $v_i$ is the binary value of the bitstring $v$ on the $i$-th index, such that $v \in \{0, 1\}^n$. Here, we also ignore the factor of two for simplicity that appears in Eq.~\eqref{eq:probs-sqrbm1-xz}, as it will cancel-out during normalization. Finally, the probabilities of the model are given with the normalization:
\begin{equation}
    p_v = \tilde{p}_v / \sum_v \tilde{p}_v.
\end{equation}
\end{proof}

\subsection{Proof of Proposition~\ref{prop:sqrbm-grads}}
\label{proof:prop:sqrbm-grads}
\begin{proof}
Let us consider the gradient of a field term acting on the visible units denoted with $a_i$. Then, using Eq.~\eqref{eq:sqrbm-probs}, one can compute the following partial derivatives:
\begin{align}
    \partial_{a_i} \tilde{p}_v &= -(-1)^{v_i} \tilde{p}_v \\    
    \partial_{a_i} \log \tilde{p}_v &= \frac{\partial_{a_i} \tilde{p}_v }{\tilde{p}_v} = \frac{-(-1)^{v_i} \tilde{p}_v}{\tilde{p}_v} = -(-1)^{v_i} \\
    \partial_{a_i} \log Z &= \frac{\partial_{a_i} \sum_v \tilde{p}_v}{\sum_v \tilde{p}_v} =  \frac{\sum_v -(-1)^{v_i} \tilde{p}_v}{\sum_v \tilde{p}_v} 
    = \sum_v -(-1)^{v_i} p_v.
\end{align}
Then, these identities can be inserted into the derivative of the loss function (see Eq.~\eqref{eq:loss_function}) to obtain:
\begin{align}
    \partial_{a_i}{\mathcal{L}} &= - \partial_{a_i} \sum_v q_v \log \tilde{p}_v / Z  \nonumber \\
    &= - \sum_v q_v \partial_{a_i} \log \tilde{p}_v + \sum_v q_v \partial_{a_i} \log Z \nonumber \\
    &= \sum_v q_v (-1)^{v_i} - \sum_v q_v \left(\sum_v (-1)^{v_i} p_v \right) \nonumber \\  
    &= \sum_v q_v \left( (-1)^{v_i} - \sum_v (-1)^{v_i} p_v \right).
\end{align}

Next, let us consider the gradient of a field term acting on the hidden units with Pauli-$k$ denoted with $b_i^k$. Then, using Eq.~\eqref{eq:sqrbm-probs}, one can compute the following partial derivatives:
\begin{align}
\partial_{b_j^k}  ||\Phi_j(v)||_2 &= \frac{\phi_j^k(v)}{||\Phi_j(v)||_2} \\
\partial_{b_j^k}  \cosh \left( ||\Phi_j(v)||_2 \right) &= \frac{\phi_j^k(v)}{||\Phi_j(v)||_2} \sinh \left( ||\Phi_j(v)||_2 \right) \\
\partial_{b_j^k}  \tilde{p}_v &= \frac{\phi_j^k(v)}{||\Phi_j(v)||_2} \tanh \left( ||\Phi_j(v)||_2 \right) \tilde{p}_v \\
\partial_{b_j^k}  \log Z &= \frac{\partial_{b_i} \sum_v \tilde{p}_v}{\sum_v \tilde{p}_v} = \sum_v \frac{\phi_j^k(v)}{||\Phi_j(v)||_2} \tanh \left( ||\Phi_j(v)||_2 \right) p_v.
\end{align}
Then, these identities can be inserted into the derivative of the loss function (see Eq.~\eqref{eq:loss_function}) to obtain:
\begin{align}
    \partial_{\theta_j^{k}}{\mathcal{L}} &= - \sum_v q_v\frac{\Phi_j^k(v)}{||\Phi_j(v)||_2} \tanh \left( ||\Phi_j(v)||_2 \right) + \sum_v  q_v \left( \sum_v \frac{ \Phi_j^k(v)}{||\Phi_j(v)||_2} \tanh \left( ||\Phi_j(v)||_2 \right) p_v \right) \nonumber \\
    &= - \sum_v q_v \left( \frac{\Phi_j^k(v) }{||\Phi_j(v)||_2} \tanh \left( ||\Phi_j(v)||_2 \right) - \sum_v \frac{ \Phi_j^k(v)}{||\Phi_j(v)||_2} \tanh \left( ||\Phi_j(v)||_2 \right)  p_v  \right).
\end{align}

Last but not least, let us consider the gradient of an interaction term acting on a visible and a hidden unit with $Z \otimes k$ denoted with $w_{i,j}^{Z,k}$. Then, using Eq.~\eqref{eq:sqrbm-probs}, one can compute the following partial derivatives:
\begin{align}
    \partial_{w_{i,j}^{Z,k}}  ||\Phi_j(v)||_2 &= \frac{(-1)^{v_i} \phi_j^k(v)}{||\Phi_j(v)||_2}\\
    \partial_{w_{i,j}^{Z,k}}  \tilde{p}_v &= \frac{(-1)^{v_i} \phi_j^k(v)}{||\Phi_j(v)||_2} \tanh \left( ||\Phi_j(v)||_2 \right) \tilde{p}_v.
\end{align}
Then, these identities can be inserted into the derivative of the loss function (see Eq.~\eqref{eq:loss_function}) to obtain:
\begin{align}
    \partial_{w_{i,j}^{Z,k}}{\mathcal{L}} = - \sum_v q_v \left( (-1)^{v_i}\frac{\phi_j^k(v) }{||\Phi_j(v)||_2} \tanh \left( ||\Phi_j(v)||_2 \right) - \sum_v (-1)^{v_i}\frac{\phi_j^k(v) }{||\Phi_j(v)||_2} \tanh \left( ||\Phi_j(v)||_2 \right)  p_v  \right).
\end{align}
\end{proof}

\subsection{Proof of Proposition~\ref{prop:sqbm-grads}}
\label{proof:prop:sqbm-grads}
\begin{proof}
Recall the following expression for the gradients of a model with Hamiltonian $H$ and target distribution $q$:
\begin{align}
    \partial_{\theta_i}{\mathcal{L}} = -\sum_v q_v \left( \frac{\Tr[\Lambda_v \partial_{\theta_i}{e^{-H}}]}{\Tr[\Lambda_v e^{-H_i}]} - \frac{\Tr[\partial_{\theta_i}{e^{-H}}]}{\Tr[e^{-H}]} \right).
\end{align}
Let us insert the power series of $e^{-H}$ to compute $\Tr[\partial_{\theta_i}{e^{-H}}]$:
\begin{align}
    \Tr[\partial_{\theta_i}{e^{-H}}] &= \Tr[e^{-H}\left(-H_i-\frac{1}{2}\left[ H, H_i \right]-\frac{1}{6}\left[ H,\left[ H, H_i \right]\right] + \cdots \right)] \nonumber \\
    &= -\Tr[e^{-H}H_i]-\frac{1}{2}\Tr[e^{-H}\left[H, H_i\right]]-\frac{1}{6}\Tr[e^{-H}\left[H,\left[H, H_i\right]\right]] + \cdots
\end{align}
Then, we observe that all the traces that include the commutator result in zero. Let us write it explicitly for the first term
\begin{align}
    \Tr[e^{-H}\left[H, H_i\right]] &= \Tr[e^{-H} H H_i]-\Tr[e^{-H} H_i H] \nonumber \\
    &=\Tr[H e^{-H} H_i]-\Tr[e^{-H} H_i H] \nonumber \\
    &=\Tr[e^{-H} H_i H]-\Tr[e^{-H} H_i H] = 0.
\end{align}

A similar reduction is possible for the term $\Tr[\Lambda_v \partial_{\theta_i}{e^{-H}}]$. Recall that Definition~\ref{def-sq-bm} assumes the commutation of the model Hamiltonian with the projector such that $[\Lambda_v, H]=0$. Then, we observe
\begin{align}
\Tr[\Lambda_v e^{-H}\left[ H, H_i \right]] &= \Tr[\Lambda_v e^{-H}H H_i] - \Tr[\Lambda_v e^{-H} H_i H]  \nonumber \\
&= \Tr[H \Lambda_v e^{-H} H_i] - \Tr[\Lambda_v e^{-H} H_i H] \nonumber \\
&= \Tr[\Lambda_v e^{-H} H_i H] - \Tr[\Lambda_v e^{-H} H_i H] = 0.
\end{align}

Since both terms in the gradient expression have a non-zero contribution from only the first term of the power series, the expression takes the following simple form:
\begin{align}
    \partial_{\theta_i}{\mathcal{L}} = -\sum_v q_v \left( -\frac{\Tr[\Lambda_v e^{-H} H_i]}{\Tr[\Lambda_v e^{-H_i}]} + \frac{\Tr[e^{-H} H_i]}{\Tr[e^{-H}]} \right).
\end{align}
\end{proof}

\subsection{Proof of Theorem~\ref{theorem:units}}
\label{proof:theorem:equivalance}
\begin{proof}
Let us consider an $\sqRBM_{n,1}$ with $n$ visible and a single hidden unit with the operator pool $\Wh = \{X,Y,Z \}$. Let us set the parameters of the field terms that act on the visible units to zero only for simplicity. Then, following Eq.~\eqref{eq:sqrbm-probs}, the unnormalized output probability of this model can be written as 
\begin{equation}
\tilde{p}_v = \cosh(\sqrt{\phi_1^X(v)^2+\phi_1^Y(v)^2+\phi_1^Z(v)^2}),
\label{eq:proof:sqrbmnn1-probs1}
\end{equation}
where we also omit constant factors, since they cancel each other during normalization. Then, using Corollary~\ref{corollary:cosh}, for small values of all parameters, the output probabilities can be written as
\begin{equation}
\tilde{p}_v \simeq \cosh(\phi_1^X(v)) \cosh(\phi_1^Y(v))\cosh(\phi_1^Z(v)).
\label{eq:proof:sqrbmnn1-probs2}
\end{equation}

Next, let us write the unnormalized output probability of an $\RBM_{n,3}$ with $n$ visible and $3$ hidden units with the same assumptions
\begin{equation}
\tilde{p}_v = \cosh(\phi_1^Z(v)) \cosh(\phi_2^Z(v))\cosh(\phi_3^Z(v)).
\label{eq:proof:rbmnn3-probs2}
\end{equation}
Then, it can be seen that Eq.~\eqref{eq:proof:sqrbmnn1-probs2} and Eq.~\eqref{eq:proof:rbmnn3-probs2} are in the same form. Consequently, one can find an equivalence with the following solution:
\begin{equation}
\phi_1^Z(v) \simeq \phi_1^Z(v), \quad \phi_2^Z(v) \simeq \phi_1^X(v), \quad \phi_3^Z(v) \simeq \phi_1^Y(v),
\end{equation}
which leads to 
\begin{equation}
    b_1^Z \simeq b_1^Z, \quad w_{n,1}^Z \simeq w_{n,1}^Z, \quad b_2^Z \simeq b_1^X, \quad w_{n,2}^Z \simeq w_{n,1}^X, \quad b_3^Z \simeq b_1^Y, \quad w_{n,3}^Z \simeq w_{n,1}^Y,
\end{equation}
and all visible unit parameters ($\boldsymbol{a}$), which we have omitted so far, can be directly considered to be equal without loss of generality. Then, by setting all parameters in $\sqRBM_{n,1}$ to some other parameter of $\RBM_{n,3}$, one obtains approximately the same output probability distribution.

For larger parameter values, the equivalence in the expressivity will still hold as Eq.~\eqref{eq:proof:sqrbmnn1-probs1} forms the product of three $\cosh$ functions with corrections. However, the direct correspondence of the parameters would no longer hold, and a system of equations has to be solved in order to find a parameter mapping.

The result trivially generalizes to any non-commuting operator set $\Wh$ as it determines the number of $\cosh$ products one can obtain. Hence, the equivalence statement holds for $\sqRBM_{n,1}$ and $\sqRBM_{n,|\Wh|}$.
\end{proof}

\newpage
\section{Numerical behavior of  gradients}\label{app:numerics-grads}

To address potential caveats in training, we report the variance of gradients for $\RBM$ and $\sqRBM$ in Figure~\ref{fig:var-grads-parity} using Proposition~\ref{prop:sqrbm-grads}. Using the parity dataset, we compute the gradients for 10000 independent random initializations from a uniform distribution between $[-10, 10]$. We report values for three types of parameters:
\begin{itemize}
    \item $a_1$ is the field term on the first visible unit 
    \item $b_1$ is the field term on the first hidden unit and can have variations such as $b_1^X$, $b_1^Y$ and $b_1^Z$ depending on the model. 
    \item $w_{1,1}$ is the interaction term on the first visible and first hidden unit and can have variations such as $w_{1,1}^{Z,X}$, $w_{1,1}^{Z,Y}$ and $w_{1,1}^{Z,Z}$ depending on the model.
\end{itemize}

We report values for $n \in \{4,6,8,10,12\}$ and $m \in \{1,2,3,4,5,6\}$, where colors and markers denote the Pauli type ($X$, $Y$ or $Z$). Results for each value of $n$ are plotted with the same color and markers. Since the variability is small, most of the lines and markers overlap with each other on the same panel.

The numerical results indicate that for each model, with a fixed number of visible units, the variance trend remains flat as the number of hidden units increases. This behavior suggests that $\sqRBM$ gradients behave very similarly to $\RBM$, and they do not suffer from vanishing gradients.
 
\begin{figure}[!h]
    \centering
    \includegraphics[width=\linewidth]{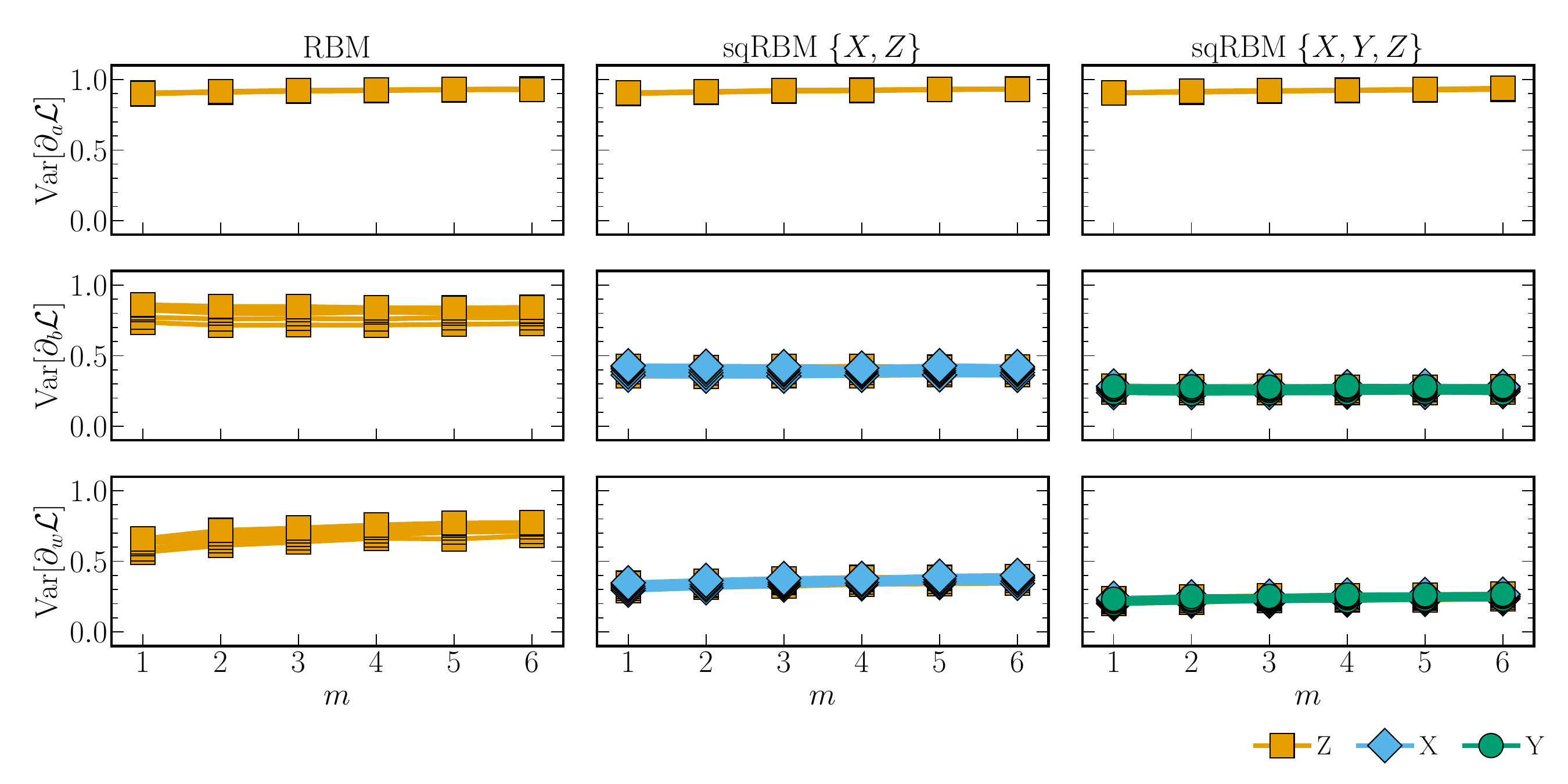}
    \caption{\textbf{Variance of gradients.} We report the variance of the gradients for the first parameter of different types of parameters ($a_1$, $b_1$, $w_{1,1}$). Each row shows results for different types of parameters, while each column shows results for different models. We report values for $n \in \{4,6,8,10,12\}$ and $m \in \{1,2,3,4,5,6\}$ and on all panels the variances show a similar behavior with very small variation with system size, as it can be observed with the overlapping lines. For models with $X$, $Y$ or $Z$ terms, we plot the results with a different marker, which also overlap.}
    \label{fig:var-grads-parity}
\end{figure}

\end{document}